\documentclass[twoside,leqno,twocolumn]{article}

\usepackage[letterpaper]{geometry}
\usepackage{siamproceedings}

\usepackage[T1]{fontenc}
\usepackage{amsfonts}
\usepackage{graphicx}
\usepackage{epstopdf}
\usepackage{enumitem}
\usepackage{algorithmic}
\ifpdf
 \DeclareGraphicsExtensions{.eps,.pdf,.png,.jpg}
\else
 \DeclareGraphicsExtensions{.eps}
\fi

\newsiamremark{remark}{Remark}
\newsiamthm{claim}{Claim}

\usepackage{amsopn}

\usepackage{listings}
\lstdefinestyle{kernel}{
 basicstyle=\ttfamily\scriptsize,
 numbers=none,
 keepspaces=true,
 columns=fullflexible,
}
\usepackage{float}
\floatstyle{plain}
\newfloat{listing}{tbp}{lol}
\floatname{listing}{Listing}

\usepackage{standalone}
\usepackage{tikz}
\usetikzlibrary{positioning,arrows.meta,fit,backgrounds,calc,matrix}

\newcommand{\sz}[1]{\$}
\usepackage{amsmath,amssymb}
\usepackage{xcolor}
\usepackage{booktabs}
\usepackage{xurl} 

\usepackage{cleveref}
\crefname{listing}{Listing}{Listings}
\Crefname{listing}{Listing}{Listings}
\crefname{equation}{eq.}{eqs.}
\Crefname{equation}{Eq.}{Eqs.}

\begin{document}

\newcommand\relatedversion{}

\title{\Large Streaming Right Multiplication over Grammar-Compressed Matrices:\\
 A Memory-Bounded GPU Engine for Genotype and Graph Data\relatedversion}

\author{Francesco Tosoni\thanks{Sant'Anna School of Advanced Studies, L'EMbeDS,
 p.zza Martiri della Libert\`a 33, 56127 Pisa PI, Italy
 (\email{Francesco.Tosoni@santannapisa.it}).}
 \and Gabriele Mencagli\thanks{Department of Computer Science, University of Pisa,
 Largo B. Pontecorvo 3, 56127 Pisa PI, Italy
 (\email{gabriele.mencagli@unipi.it}).}}


\date{}

\maketitle

\fancyfoot[R]{\scriptsize{Copyright \textcopyright\ 2027\\
Copyright for this paper is retained by authors}}

\begin{abstract}
Grammar-compressed matrices store a matrix's non-zero structure as a RePair straight-line program (SLP), supporting matrix--vector products in time and space proportional to the {\em compressed} size. We target the regime where this is decisive on a GPU: when the uncompressed matrix exceeds device memory, so footprint (not floating-point throughput) is the binding constraint. Our SLP is a directed acyclic graph (DAG) of out-degree~2, and the right product $y=Mx$ is a single bottom-up sweep (leaves~$\to$~roots): a conflict-free {\em gather}. We make the grammar {\em properly layered} (every nonterminal child one level below its parent) via {\em pass-through completion}, which inserts identity nodes to carry values upward until consumed. This yields a {\em streaming} evaluation in which each level reads only the level below and writes the next, so the live set fits in two alternating read-only/write-only buffers instead of scaling with the whole grammar; the per-level width equals the live set. On genotype matrices, where a polygenic score is exactly the right product $y=G\beta$, a CUDA implementation shows a clear {\em space} advantage: a device footprint 3.5 to 8 times smaller than a materialized cuSPARSE CSR baseline, single-vector times within a small factor of cuSPARSE, and consistently lower GPU energy. Because the sweep needs only an associative and commutative combine, the same engine and schedule evaluate any homomorphism into a commutative monoid over the grammar by swapping a small leaf/combine/emit policy; the same reachability sweep then scales to a billion-edge Software Heritage graph ($261$ TB dense and unmaterializable, $21\times$ smaller serialized than cuSPARSE's CSR), where the memory argument holds. We frame this as an algorithm-engineering case study: structural metrics (depth, live-set width, completion cost) are measured, architecture-independent grammar properties, whereas time and energy are profiled on a single board.
\end{abstract}

\noindent\textbf{Artifact.} Code: \href{https://doi.org/10.5281/zenodo.23139248}{doi:10.5281/zenodo.23139248} (development version: \url{https://github.com/ftosoni/g-mm-repair}); data: \href{https://doi.org/10.5281/zenodo.22677747}{doi:10.5281/zenodo.22677747}. The code archive includes every raw log and the script that regenerates the tables and figures, plus supporting results omitted here for space (per-matrix host construction cost; cuSPARSE SpMM algorithm sweep).



\section{Introduction}
\label{sec:intro}

Following CLA~\cite{cla_vldbj,cla_cacm}, research into computation-friendly matrix formats~\cite{impossibility_gra_linalg,accelerating_gnn,arroyuelo2026boosting,rpq-k2tree-vldbj,aware,unified_lossless_sparse,mmrepair,francisco2022friendly,efficient_marino,greener} enables linear algebra directly on compressed representations. Similar principles drive compressed-text engines on GPUs and FPGAs~\cite{gtadoc,ftadoc} and indexing repetitive collections~\cite{navarro-rep-i,navarro-rep-ii}, where space scales with repetitiveness rather than raw size, as for the $r$ measure of the r-index~\cite{rindex_soda,rindex_jacm} or run-length-BWT machinery~\cite{boucher-uk-biobank} adapted for genotype data.

Contemporaneously with CLA~\cite{cla_vldbj}, Tabei {\em et al.}~\cite{tabei-grammar-compressed-pls} grammar-compressed binary data matrices with RePair~\cite{repair}, computing the matrix--vector products of partial least squares regression through row and column access to the compressed form. The \texttt{mm-repair} scheme~\cite{mmrepair} handles general real-valued matrices: it turns a matrix $M\in\mathbb{R}^{n\times m}$ into a string, writing each non-zero row by row as a symbol $\langle\text{value},\text{column}\rangle$, and grammar-compresses it with RePair~\cite{repair}, which repeatedly replaces a frequent pair of adjacent symbols by a new rule. Row fragments that recur across rows thus become shared rules: the grammar is a DAG whose leaves are the non-zeros and whose roots spell the rows. The right product $y=Mx$ needs no decompression: each leaf evaluates to its value times the matching entry of $x$, each rule sums its two children, and each row root yields an entry of $y$. Evaluating children before parents computes each rule's partial sum once, however often the rule recurs, in $O(\vert{}\mathcal{C}\vert{}+\vert{}\mathcal{R}\vert{})$ time and $O(\vert{}\mathcal{R}\vert{})$ space~\cite{mmrepair}, where $\mathcal{R}$ is the rule set and $\mathcal{C}$ the compressed sequence of rows; the left product $x^{\mathsf{T}}=y^{\mathsf{T}}M$ is the analogous parents-before-children pass. Prior works~\cite{mmrepair,greener} parallelize this by partitioning $M$ into row blocks, compressing them separately, and assigning each to a thread.

We propose an orthogonal, GPU-specific parallelization \emph{within} a single grammar. We focus on \emph{memory-constrained engineering}: maximizing \emph{in-place} computation near physical memory limits. Nodes at the same depth never depend on each other, so the bottom-up pass runs level by level, with all nodes of a level running in parallel; it is a conflict-free \emph{gather}, since each node only reads its children. The dual left product is a \emph{scatter} with write contention on highly shared nodes: an obstacle on GPUs.

This up- and down-sweep pattern recalls Blelloch’s parallel prefix sum~\cite{BlellochTR90}: our right-product gather mirrors his reduction, while his distribution mirrors the left-product scatter. Yet, though balanced trees allow local rewrites, our DAG’s high-in-degree nodes force accumulation under contention. \emph{We thus restrict this paper to the right product}: it is essential for power iteration (e.g., PageRank)~\cite{greener,francisco2022friendly} and avoids scatter bottlenecks.

\paragraph{Streaming on a properly layered grammar.}

The sweep is embarrassingly parallel \emph{within each level}. Still, the target regime needs more than node independence: a \emph{properly layered} grammar, where every nonterminal child sits one level below its parents. Then each level reads only the level below and writes the next, so, in the bottom-up traversal, every frontier is \emph{freed} once consumed, and two buffers alternate read- and write-only. We retain RePair's compression capabilities and impose proper layering via \emph{pass-through completion} (\cref{sec:through}), inserting identity nodes that propagate values upward until they are consumed; the resulting per-level width precisely matches the live set (i.e., the number of active nonterminals).

\paragraph{Relation to known techniques.}

Several ingredients of our engine have precedents in prior work. A straight-line program is a circuit, i.e., a DAG~\cite{lohrey-slp-survey}; longest-path leveling is the level schedule of parallel sparse triangular solves~\cite{anderson-saad}, which on GPUs also launches one kernel per level (or per chain of narrow levels)~\cite{naumov-sptrsv}; splitting long edges by dummy vertices to make a layered DAG \emph{proper} is the first step of layered graph drawing~\cite{sugiyama}; and alternating input/output frontiers are standard in level-synchronous GPU traversal~\cite{scalable_gpu_graph_traversal,naumov-sptrsv}. G-TADOC~\cite{gtadoc} already traverses grammar DAGs on GPUs for text analytics, scheduling rules dynamically and allocating per-rule buffers over the whole grammar. Our contribution is their combination for grammar-compressed linear algebra on the GPU, and the guarantee it yields: completion makes every level width \emph{equal} to the live set, so working memory is bounded by a measurable, \emph{a priori} property of the grammar rather than by its size. We engineer this into a branch-free GPU sweep with shared pass-through chains and on-the-spot emission of row roots, generalize it to commutative-monoid homomorphisms, and quantify its cost (depth, width, completion inflation) on real genotype and graph data.

\paragraph{Scope and non-goals.}

We pursue a ``compress-once, multiply-many-times'' scenario in which the input vectors are not known \emph{a priori}. The GPU is our sole target architecture; a host-constructed grammar feeds a level-synchronous GPU engine, which then collects all throughput metrics. We evaluate performance against the 16-thread CPU \texttt{mm-repair} baseline from~\cite{mmrepair} and the NVIDIA cuSPARSE CSR kernel \cite{nvidia2026cusparse} as references. The depth, live-set width, and structural sizing metrics are architecture-independent grammar properties computed on the host during construction; these \emph{a priori} grammar features determine the engine's device footprint and launch count before any GPU run. We exclude left multiplication and FPGA-based deployments and defer them to future work. We delineate in \cref{sec:setup,sec:limitations} board-dependent versus invariant metrics, and validate on two domains stressing different axes: genotype matrices, whose dense form is prohibitive at biobank scale so $y=G\beta$ wins on space at time parity, and repetitive graphs, where the same monoid sweep reaches a billion-edge software graph.

\section{Background: grammar-compressed matrices}
\label{sec:background}

\paragraph{CSRV.}

We use the CSRV (compressed sparse row/value) representation of~\cite{mmrepair}, a simple variant of CSR. Let $V$ store the distinct non-zero values of $M\in\mathbb{R}^{n\times m}$. Scanning $M$ row by row, each non-zero entry $M[i][j]=V[\ell]$ is emitted as a pair $\langle\ell,j\rangle$. A delimiter \$ terminates each row, thereby producing a string $S$ of value-column pairs (the \emph{terminals}) and \$. \Cref{eq:M,eq:S} give the small matrix used as our running example (an $8\times7$ matrix with $|V|=4$ distinct non-zero values and three distinct rows, repeated to expose shared structure) and its CSRV string $S$; the value array is $V=[ 1.2,\ 2.3,\ 3.4,\ 4.5 ]$. In \cref{eq:S} we write out only the three distinct rows ($1,2,6$); the $\cdots$ mark where the repeated row blocks recur: rows $3,5$ repeat row~$2$, rows $4,8$ repeat row~$1$, and row~$7$ repeats row~$6$. The pair $\langle\ell,j\rangle$ encodes value $V[\ell]$ in column $j$; for example, $\langle2,1\rangle$ is $V[2]=2.3$ in column~$1$, while the same value in column~$6$ is $\langle2,6\rangle$. Only equal values in the \emph{same} column share a pair. A row delimiter \$ closes each row.

{\setlength{\arraycolsep}{3pt}
\begin{equation}\label{eq:M}
M=\begin{pmatrix}
2.3 & 0 & 1.2 & 1.2 & 1.2 & 2.3 & 1.2\\
1.2 & 4.5 & 3.4 & 1.2 & 1.2 & 2.3 & 4.5\\
1.2 & 4.5 & 3.4 & 1.2 & 1.2 & 2.3 & 4.5\\
2.3 & 0 & 1.2 & 1.2 & 1.2 & 2.3 & 1.2\\
1.2 & 4.5 & 3.4 & 1.2 & 1.2 & 2.3 & 4.5\\
0 & 4.5 & 3.4 & 2.3 & 2.3 & 0 & 4.5\\
0 & 4.5 & 3.4 & 2.3 & 2.3 & 0 & 4.5\\
2.3 & 0 & 1.2 & 1.2 & 1.2 & 2.3 & 1.2
\end{pmatrix}
\end{equation}}
\begin{equation}\label{eq:S}
\small
\begin{aligned}
S={} & \langle2,1\rangle\langle1,3\rangle\langle1,4\rangle\langle1,5\rangle\langle2,6\rangle\langle1,7\rangle \$ \\
     & \langle1,1\rangle\langle4,2\rangle\langle3,3\rangle\langle1,4\rangle\langle1,5\rangle\langle2,6\rangle\langle4,7\rangle \$\ \cdots \\
     & \langle4,2\rangle\langle3,3\rangle\langle2,4\rangle\langle2,5\rangle\langle4,7\rangle \$\ \cdots
\end{aligned}
\end{equation}

\paragraph{Grammar.} 

RePair~\cite{repair} compresses $S$ by repeatedly replacing the most frequent pair of adjacent symbols with a new nonterminal, modified to ensure it never pairs elements across a \$ boundary~\cite{mmrepair}. This process yields a set of production rules $N_i\to A_iB_i$ (where each $A_i,B_i$ is either a terminal $\langle\ell,j\rangle$ or a preceding nonterminal $N_{<i}$) and a final sequence $\mathcal{C}$ whose constituent symbols expand directly to the matrix rows. The indexing follows a strict topological ordering: $h<i$ whenever $N_h$ appears on the right-hand side of the rule for $N_i$. Executing RePair on the string $S$ of \cref{eq:S} yields:
\begin{equation}
\label{eq:rules}
\small
\mathcal{R}=\left\{
\begin{array}{ll}
N_1\!\to\!\langle1,4\rangle\langle1,5\rangle &
N_2\!\to\!N_1\langle2,6\rangle \\[2pt]
N_3\!\to\!\langle4,2\rangle\langle3,3\rangle &
N_4\!\to\!\langle1,1\rangle N_3\\[2pt]
N_5\!\to\!\langle1,3\rangle N_2 &
N_6\!\to\!\langle2,1\rangle N_5 \\[2pt]
N_7\!\to\!N_2\langle4,7\rangle &
N_8\!\to\!N_4 N_7\\[2pt]
N_9\!\to\!N_6\langle1,7\rangle &
N_{10}\!\to\!\langle2,4\rangle\langle2,5\rangle \\[2pt]
N_{11}\!\to\!N_3 N_{10} &
N_{12}\!\to\!N_{11}\langle4,7\rangle
\end{array}\right\}
\end{equation}
\begin{equation}
\label{eq:C}
\mathcal{C}=N_9 \$ N_8 \$ N_8 \$ N_9 \$ N_8 \$ N_{12} \$ N_{12} \$ N_9 \$ .
\end{equation}
In this specific instance, each row within $\mathcal{C}$ reduces to a single nonterminal symbol; thus, $\mathcal{C}$ holds one root per row (eight in all), drawn from the set $\{N_9,N_8,N_{12}\}$. In practice, $\mathcal{C}$ can be longer and may incorporate bare terminals. Indeed, RePair stops when no more pairs of consecutive symbols appear more than once~\cite{mmrepair}. The deepest root, $N_9$, encapsulates a chain of five nested rules ($N_9\!\to\!N_6\!\to\!N_5\!\to\!N_2\!\to\!N_1$), indicating that the grammar spans five logical levels.

\paragraph{The right-product sweep.} Let $\mathrm{eval}_x(\langle\ell,j\rangle)=V[\ell]\cdot x[j]$ denote the scalar contribution of an individual matrix entry to the output product. The underlying grammar transforms the multiplication $y=Mx$ into a single bottom-up evaluation pass governed by three key properties established in~\cite{mmrepair}.

\begin{lemma}[Additivity of $\mathrm{eval}_x$~\cite{mmrepair}]
\label{lem:eval}
For any grammar rule $N_i\to A_iB_i$, the evaluation distributes additively: $\mathrm{eval}_x(N_i)=\mathrm{eval}_x(A_i)+\mathrm{eval}_x(B_i)$.
\end{lemma}

\begin{lemma}[Rows as roots~\cite{mmrepair}]
\label{lem:rows}
Given a compressed sequence $\mathcal{C}=N_{i_1}\$\cdots N_{i_n}\$$ and a product $y=Mx$, the output vector matches the root evaluations: $y[r]=\mathrm{eval}_x(\mathrm{root}_r)$, where $\mathrm{root}_r=N_{i_r}$, for every row index $r=1,\dots,n$.
\end{lemma}

\begin{theorem}[Compressed-time right product~\cite{mmrepair}]
\label{thm:right}
Given the grammar components $(\mathcal{C},\mathcal{R},V)$ representing a matrix $M\in\mathbb{R}^{n\times m}$ and an input vector $x\in\mathbb{R}^m$, the product $y=Mx$ can be evaluated in $O(|\mathcal{C}|+|\mathcal{R}|)$ time using $O(|\mathcal{R}|)$ words of auxiliary storage.
\end{theorem}

In practice, the CPU implementation of \cite{mmrepair} populates an auxiliary array $W[1,|\mathcal{R}|]$ such that $W[i]=\mathrm{eval}_x(N_i)$ via a single forward traversal. Because the topological numbering ensures child nodes precede their respective parents, \cref{lem:eval} resolves sequentially as $W[i]=\mathrm{cv}(A_i)+\mathrm{cv}(B_i)$, where $\mathrm{cv}(N_h)=W[h]$ and $\mathrm{cv}(\langle\ell,j\rangle)=V[\ell]x[j]$. The final answers are then extracted via $y[r]=W[i_r]$, as in \cref{lem:rows}. Preserving the explicit \$ boundaries enforces a clean, row-separated representation that underpins the right-product execution; all our structural GPU-specific modifications respect this property. \Cref{thm:right} provides the operational bound our engine maintains, and the streaming framework in \cref{sec:gpu} is its level-synchronous, parallel GPU realization.

\section{The grammar as a DAG, and its levels}
\label{sec:dag}

We interpret the rules in \cref{eq:rules} as a directed acyclic graph $G$: each nonterminal maps to an internal node, each distinct terminal forms a leaf node, and directed edges span from every parent $N_i$ to its two children. Consequently, nonterminals maintain an out-degree of 2, terminals remain childless leaves, and the set of roots corresponds to the symbols active in $\mathcal{C}$. Crucially, $G$ forms a \emph{DAG rather than a tree}, as reused sub-expressions exhibit an in-degree $>1$. For instance, $N_2$ drives both $N_5$ and $N_7$, $N_3$ feeds both $N_4$ and $N_{11}$, and individual roots like $N_9,N_8,N_{12}$ serve multiple rows within $\mathcal{C}$. From the perspective of the right product, this multi-parent sharing is benign because it represents pure read-sharing.

\paragraph{Levels.} We assign every node an integer level corresponding to its maximal path distance from a leaf:
\begin{equation}
\label{eq:level}
\begin{aligned}
\mathrm{lvl}(\text{terminal})&=0,\\
\mathrm{lvl}(N_i\!\to\!A_iB_i)&=1+\max\bigl(\mathrm{lvl}(A_i),\mathrm{lvl}(B_i)\bigr).
\end{aligned}
\end{equation}
Applying this definition to the rules in \cref{eq:rules} generates five distinct strata: $\{N_1,N_3,N_{10}\}$, $\{N_2,N_4,N_{11}\}$, $\{N_5,N_7,N_{12}\}$, $\{N_6,N_8\}$, and $\{N_9\}$, as illustrated in \cref{fig:dag}. The left spine $N_9\!\to\!N_6\!\to\!N_5\!\to\!N_2\!\to\!N_1$ realizes the five levels. Nodes are \emph{shared} ($N_2$ feeds $N_5$ and $N_7$; $N_3$ feeds $N_4$ and $N_{11}$; the terminal $\langle4,7\rangle$ feeds $N_7$ and $N_{12}$). Many edges skip a level by reaching a \emph{terminal} (e.g., $N_9,N_6,N_5$ each drop to a level-0 leaf): these are harmless, as terminals read the persistent input and are never freed. The critical edges are the \emph{nonterminal} level-skipping edges, such as \textcolor{red!70!black}{$N_8\!\to\!N_4$} (red) from level~4 to level~2. This skips a level and violates proper layering, since evaluating $N_8$ would require a value two frontiers below it. \emph{Pass-through completion} (\cref{fig:dag}) splits exactly such edges so the streaming sweep applies.

\begin{lemma}[Levels are antichains]
\label{lem:antichain}
If $\mathrm{lvl}(u)=\mathrm{lvl}(v) $ for distinct nodes $u\neq v$, then $G$ contains no directed edge connecting $u$ and $v$.
\end{lemma}
\begin{proof}
Every valid edge in $G$ terminates at a child node possessing a strictly smaller level according to \cref{eq:level}; hence, the source and destination endpoints of an edge can never occupy the same level.
\end{proof}

\Cref{lem:antichain} serves as the core parallelization driver for our architecture: \emph{all nodes residing within a given level can be evaluated simultaneously}. The total number of levels (i.e., the overall grammar depth) $L=\max_i\mathrm{lvl}(N_i)$ defines the critical path of the computation and determines the number of sequential GPU kernel launches. Yet, node independence alone does not fully satisfy our design objectives; we require a more rigid structural configuration.

\begin{definition}[Proper layering]
\label{def:proper}
A grammar is \emph{properly layered} if every \emph{nonterminal} child node is positioned exactly one level beneath its parent node: for any rule $N\to AB$ where $\mathrm{lvl}(N)=k$, any child that is a nonterminal must satisfy $\mathrm{lvl}=k-1$ (the \emph{proper hierarchies} of~\cite{sugiyama}, relaxed for terminals). Terminal children face no such constraints; they interface directly with the persistent inputs $x$ and $V$, which remain globally resident in device memory and are never deallocated.
\end{definition}

\begin{lemma}[Liveness / double buffering]
\label{lem:liveness}
When executing over a properly layered grammar, the right-product sweep requires only two active value buffers in working memory. While evaluating level $k$, the kernel reads exclusively from level $k-1$ and writes to level $k$. Upon completing level $k$, the frontier at level $k-1$ has no remaining consumers and can be safely deallocated; any root node finalized at level $k$ is immediately emitted to the output vector $y$. Consequently, the transient memory footprint is bounded by $O(\max_k w_k)$, where $w_k$ is the width of level $k$, rather than scaling as $O(|\mathcal{R}|)$, and the execution alternates between two dedicated read-only and write-only buffers.
\end{lemma}
\begin{proof}
By \cref{def:proper}, all parent nodes consuming a nonterminal from level $k-1$ must reside exactly at level $k$. Once level $k$ is fully processed, all potential consumers of the level-$(k-1)$ buffer have finished execution, allowing it to be cleared. Terminal nodes are read from persistent global memory and do not depend on transient buffers.
\end{proof}

We use pass-through completion (\cref{sec:through}) to ensure proper layering, enabling streaming, double-buffered execution. Minimizing $L$ reduces kernel launches, while bounding the maximum level width $w_k$ ensures live sets fit within device memory. This approach mirrors sparse dynamic programming on small-width DAGs~\cite{makinen-sparse-dynprog}, where costs scale with the width of a minimum path cover; here, the per-level antichains of \cref{lem:antichain}.

\begin{figure}[t]
\centering
\includegraphics[width=.95\linewidth]{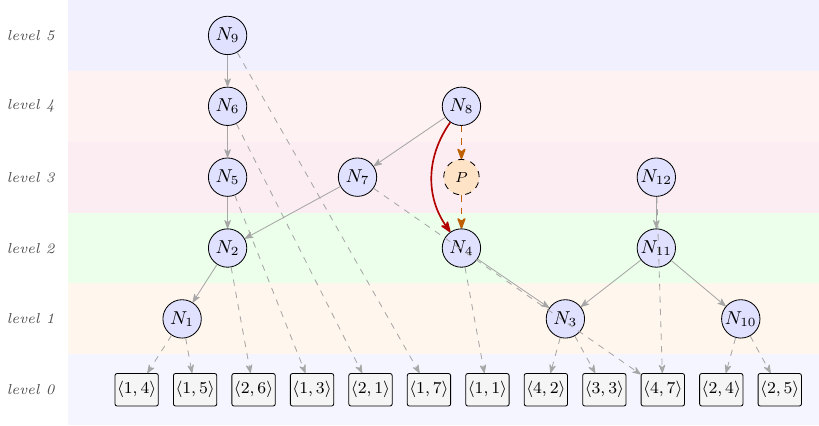}
\caption{The grammar DAG of \cref{eq:rules} ($L=5$ levels, background bands; \cref{eq:level}). Circles represent non-terminals, while boxes represent terminal leaves. Dashed gray edges are harmless terminal level-skips (\cref{def:proper}). The lone nonterminal level-skipping edge $N_8\!\to\!N_4$ (red) violates proper layering. \emph{Pass-through completion} (\cref{sec:through-layer}) splits it through the synthesized node $P$ (orange path $N_8\!\to\!P\!\to\!N_4$), making every nonterminal edge level-adjacent.}
\label{fig:dag}
\end{figure}

\section{A proper-layered streaming engine}
\label{sec:through}

Our central thesis is that we can fully preserve the compression advantages of the \texttt{mm-repair} format while enabling the streaming, double-buffered execution model of \cref{lem:liveness}. We achieve this by transforming the raw RePair output into a properly layered grammar during a post-processing phase. The overhead introduced consists of the synthesized pass-through nodes; yet, our experiments show that this cost remains a small fraction of the matrix's CSR footprint ($6$--$11\%$ on genotype matrices; \cref{sec:geno-exp}).

\subsection{Step 1: RePair (unchanged)}
We execute RePair on the input sequence $S$ exactly as described in~\cite{mmrepair}, processing the matrix row by row and strictly enforcing the \$ boundaries. We treat it here as an immutable black box. Its primary drawback for high-throughput parallelization is structural: its greedy, global selection of the most frequent digram yields deep graph topologies containing level-skipping edges (\cref{fig:dag}). We report the overall depth $L$ and the exact distribution of edge spans for our datasets in \cref{sec:through-exp}.

\subsection{Step 2: pass-through completion (layering)}
\label{sec:through-layer}
Because standard RePair offers no structural layering guarantees, we explicitly add nodes to the grammar graph. For every nonterminal directed edge $N\to\dots M\dots$ where the level differential $d=\mathrm{lvl}(N)-\mathrm{lvl}(M)$ is strictly greater than 1, we intercept the edge by inserting a sequential chain of $d-1$ \emph{pass-through} nodes denoted $M=M_0,M_1,\dots,M_{d-1}$, with $M_i$ at level $\mathrm{lvl}(M)+i$. Chains are shared, not duplicated per edge: each $M$ receives one pass-through per level up to its highest consumer, and all its parents read from that single chain. These artificial nodes implement identity rules of the form $M_i\to M_{i-1}$, ensuring that $\mathrm{eval}_x(M_i)=\mathrm{eval}_x(M_{i-1})$. The parent node $N$ is then modified to read from $M_{d-1}$ at level $\mathrm{lvl}(N)-1$; terminal edges are left unaltered (\cref{def:proper}). 

Completion rewrites only the rule set $\mathcal{R}$, never $\mathcal{C}$, preserving the row roots and hence $y[r]=\mathrm{eval}_x(\mathrm{root}_r)$; it is an offline, amortized preprocessing cost, not an online penalty (\cref{sec:through-exp}). \Cref{fig:dag} shows the result: the single level-skipping edge $N_8\!\to\!N_4$ (red) is split by one pass-through $P$ at level~3 ($P\to N_4$), so $N_8$ reads $P$ one level below and every nonterminal edge becomes level-adjacent (\cref{lem:liveness}). The pass-through is not overhead but bookkeeping: it \emph{is} $N_4$ held live across level~3 for its higher consumer $N_8$, so each band's width is exactly the live set. \Cref{fig:trace} replays this as a numeric, round-by-round computation.

Such insertions are not wasteful: a pass-through crossing level $k$ is a value \emph{alive but unconsumed} there, so the completed width $w_k$ tracks the live set and the buffers hold exactly the values in flight (on GPUs, forwarded frontier slots; on an FPGA pipeline they would correspond to pipeline registers). Worst-case completion adds $O(|\mathcal{R}| L)$ nodes, but the actual inflation follows the edge-span distribution, which we measure empirically (\cref{sec:through-exp}).

\paragraph{Depth still matters, but is not enough.}

Grammar height governs both the sequential launch count $L$ and the length of individual pass-through chains. Thus, pre-balancing to a logarithmic height bound of $O(\log n)$~\cite{ganardi,balancing-urbina} is an attractive complementary optimization. We do not apply it in this work (our pipeline consumes the raw RePair grammar unchanged), and we stress that it would not replace completion. Structural balancing bounds only the \emph{depth}, not the active \emph{live set}, meaning a balanced grammar can still host nonterminal edges that skip multiple levels. Because only explicit completion guarantees compliance with the streaming model, balancing is orthogonal to our contribution; we leave its integration to future work (\cref{sec:conclusion}).

\subsection{Step 3: the streaming, double-buffered sweep}
\label{sec:gpu}

During the offline phase, we evaluate \cref{eq:level} on the completed grammar. Nodes are laid out level by level, with rules concatenated into a flat array and partitioned by offsets so each level $k$ occupies a contiguous block. Node children are encoded as absolute offsets relative to the {\em preceding frontier array}. Maintaining nodes in RePair index order ensures that child reads remain monotone and largely contiguous across warps. Execution of $y=Mx$ follows \cref{lst:kernel} using alternating buffers, $\mathit{prev}$ and $\mathit{cur}$. For each level $k=1,\dots,L$, a single kernel launch reads level $k-1$ and writes level $k$, with the host swapping buffer pointers to alternate modes (\cref{lem:liveness}). The maximum transient memory footprint is $O(\max_k w_k)$.

To prevent warp divergence, we handle all node variants using a branch-free, fused multiply--add sequence. We allocate a dedicated \emph{zero terminal} at $T[0]=0$ and encode all child references as signed array indices: a non-negative index $c$ indicates a lookup in the previous frontier buffer $\mathit{prev}[c]$, whereas a negative index $c$ routes the lookup to the persistent terminal array $T[-c-1]$, where $T[p]=V[\ell_p] x[j_p]$. Every node can then be evaluated using the unified expression $\mathit{cur}=\mathrm{coeff}\cdot\mathrm{cv}(\mathit{left})+\mathrm{cv}(\mathit{right})$, where the tuple parameters are assigned as follows: $(\mathrm{coeff},\mathit{right})=(1,B)$ for standard binary rules $N\to AB$; $(1,\text{zero})$ for pass-through operations $N\to A$; and $(t,\text{zero})$ for run-length rules $N\to A^{t}$, giving $\mathrm{eval}_x(N)=t\cdot\mathrm{eval}_x(A)$ from a single child read. The last case never arises for us: within a row, the column index is distinct at every position, so no pair (hence no digram) can repeat, and RePair emits only binary productions; the engine nonetheless supports it for run-length-enriched grammar constructions too (\cref{sec:conclusion}).

\begin{listing}[t]
\begin{lstlisting}[style=kernel]
kernel round(level k, prev, cur):  # one thread t per node of level k
 r = rules[k][t]
 cur[t] = r.coeff * cv(r.left) + cv(r.right) # one fused MAD, no branch
kernel emit(level k, cur):     # one thread per C-occurrence at level k
 atomicAdd(&y[row[e]], cur[pos[e]])      # emit on the spot, then free
# host: zero y; for k = 1..L { round(k,prev,cur); emit(k,cur); swap(prev,cur) }
\end{lstlisting}
\caption{Simplified streaming round: read $\mathit{prev}$ (level $k-1$), write $\mathit{cur}$ (level $k$); the host swaps the buffers after each launch. $\mathrm{cv}(c)=\mathit{prev}[c]$ if $c\ge 0$, else $T[-c-1]$. A node of $\mathcal{C}$ is emitted into $y$ the moment its level is computed (atomic add), then freed.}
\label{lst:kernel}
\end{listing}

\Cref{fig:sweep} illustrates the execution on three consecutive levels of the running example ($L=5$; \cref{fig:trace} runs it in full). The round kernel is a conflict-free gather: threads write distinct slots of $\mathit{cur}$, and concurrent reads of shared children in $\mathit{prev}$ are harmless. A root of $\mathcal{C}$ is emitted into $y$ by atomic addition as soon as its level is finalized, then discarded; this keeps roots from propagating upward, so $w_k$ stays minimal rather than inflating to $\Theta(\text{rows})$. Pass-through completion keeps every lookup on the adjacent level rather than scattering across the DAG. The sequential launch count is bounded by $L$, which height-balancing would cap logarithmically. Per-round width tapers toward the roots; occupancy peaks where most of the reduction happens.

\begin{figure}[h]
\centering
\includegraphics[width=.85\linewidth]{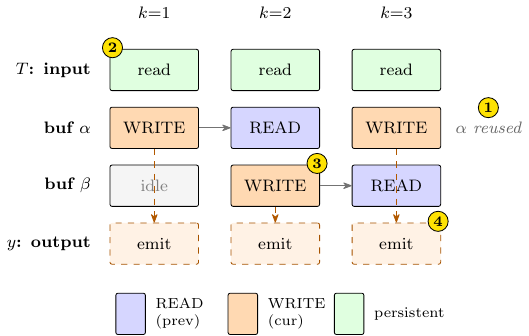}
\caption{The double-buffered, level-synchronous sweep (\cref{lst:kernel}), shown for the first three levels ($k=1,2,3$). Buffers $\alpha,\beta$ swap roles every level. \textcircled{1} double buffering; \textcircled{2} persistent terminals; \textcircled{3} branch-free fused MAD; \textcircled{4} emit-on-the-spot.}
\label{fig:sweep}
\end{figure}

\begin{figure}[t]
\centering
\includegraphics[width=.95\linewidth]{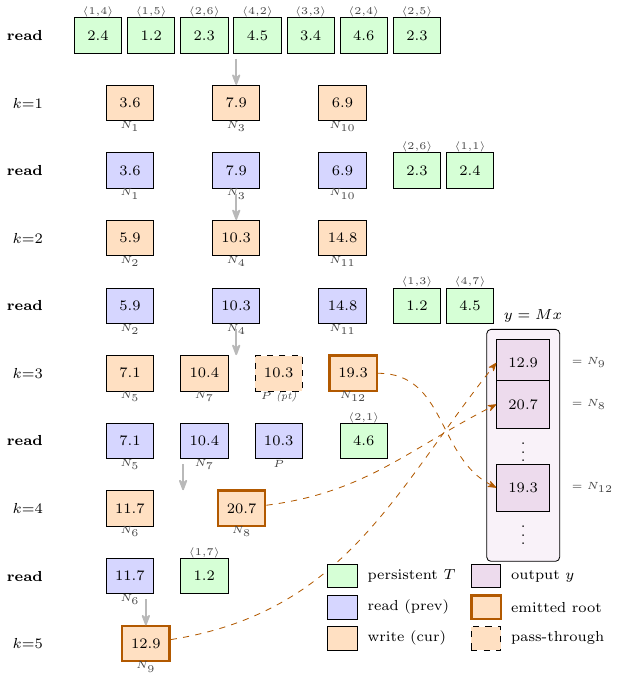}
\caption{The DAG grammar execution for input $x=[ 2,1,1,2,1,1,1 ]^{\top}$ and $V=[ 1.2,2.3,3.4,4.5 ]$. Terminals evaluate as $\langle\ell,j\rangle \mapsto V[\ell]\cdot x[j]$. Distinct row roots $N_{12},N_8,N_9$ evaluate to $19.3,20.7,12.9$. Ellipses ($\ldots$) elide repeated rows.}
\label{fig:trace}
\end{figure}

\subsection{Beyond $(+,\times)$: a monoid-homomorphism engine}
\label{sec:semiring}

Our engine does not rely on full-ring axioms, so the same schedule that computes the genotype product also evaluates graph reachability, with no change to Steps~1--3. Correctness over the shared DAG needs only that the combine $\oplus$ be \emph{associative} (a monoid): then a reused nonterminal's value is invariant to how its sub-expansion is parenthesized, so the intra-level parallelism of \cref{lem:antichain} and the liveness of \cref{lem:liveness} still hold. Inside the grammar, commutativity is not required, since RePair preserves symbol order and never pairs across row boundaries; it is required only by the final scatter of a row's roots into $y$, which is atomic and hence unordered (all three semirings below are commutative; a non-commutative $\oplus$ would need an order-preserving emission). The leaf map (non-zero terminals) is the multiplicative $\otimes$; an unmapped entry takes the $\oplus$-identity (the semiring's zero: $0$ for $(+,\times)$ and Boolean, $+\infty$ for tropical).

Our engine evaluates \emph{any} homomorphism into a commutative monoid over the grammar without schedule adjustments. Users swap low-level device primitives via a unified configuration struct (\texttt{SEMIRING=\{plustimes,boolean,tropical\}}). Two semirings are of immediate interest: the \emph{Boolean} semiring ($\oplus={}$\textsc{or}, $\otimes={}$\textsc{and}), where one matrix--vector pass expands a single BFS frontier; and the \emph{Tropical} semiring ($\oplus=\min$, $\otimes={}+$), where a pass performs one Bellman--Ford edge relaxation. A single pass is the natural unit of work, and the one we measure (\cref{sec:graph-exp}); iterating to the BFS/SSSP fixpoint requires a square, endomorphic node remap, which we leave to future work (\cref{sec:conclusion}). This recasts adjacency compression as an enabler for matrix-algebraic graph queries, such as WebGraph~\cite{webgraph,webgraph_rust}, Zuckerli's~\cite{zuckerli} intervalisation for adjacency--vector products~\cite{francisco2022friendly,greener}, $k^2$-tree representations~\cite{k2tree-spire}\cite[\S9.2.1]{navarro_book}, and regular path queries over compressed structures, via Boolean-matrix algebras~\cite{rpq-k2tree-vldbj,arroyuelo2026boosting} or ring-based indexes~\cite{rpq-compact-vldbj,time-space-rpq}.

\section{Experimental setup}
\label{sec:setup}

All metrics, including host-side grammar decomposition, structural analysis, and GPU execution throughput profiling, are collected on a device exclusively reserved for our experiments and equipped with an NVIDIA GB10 Grace–Blackwell superchip architecture. This hardware platform runs Ubuntu 24.04 LTS (Linux kernel 6.17) and features a 20-core Arm CPU (aarch64 architecture) tightly coupled to 119 GiB of unified coherent memory shared over a high-bandwidth NVLink-C2C interconnect. The onboard GPU hosts 48 streaming multiprocessors (SMs) targeting compute capability 12.1 ({\tt sm\_121}). The software stack comprises the CUDA 13.0 toolkit, g++ 13.3, and CMake 3.28. Energy is measured from the GPU's power telemetry via NVML: we sample {\tt nvmlDeviceGetTotalEnergyConsumption} on the Blackwell device before and after a sustained evaluation loop (repeated until at least 1.5 s of sustained execution) and divide the energy delta by the number of products. All energy figures (mJ per vector) are thus GPU-only, as reported by NVML: they exclude the Grace CPU and the rest of the system.

The most consequential architectural trait of the GB10 platform for our implementation is its fully unified and coherent physical memory layout: the host CPU and destination GPU share the 119 GiB pool over the NVLink-C2C bus. This eliminates host--device transfers for the active streaming frontiers (\cref{lem:liveness}), which never leave the shared pool; the base grammar is copied once into managed memory at setup, outside the timed loop. We allocate the frontier arrays via {\tt cudaMallocManaged}.

\section{Experimental results}
\label{sec:through-exp}

We first evaluate the $(+,\times)$ engine on genotype matrices (\cref{sec:geno-exp}), then its monoid generalization on repetitive graphs (\cref{sec:graph-exp}); \cref{sec:correctness} describes how every reported output is verified.

\subsection{Genotype matrices: the $(+,\times)$ product}
\label{sec:geno-exp}

We evaluate the engine on large genotype datasets using the GB10 platform described in the preceding section: computing a polygenic risk score is the right product $y=G\beta$ \cite{choi2020tutorial}. Biobank-scale panels (hundreds of thousands of samples, millions of variants) exceed physical memory, which is why genomics relies on specialized run-length/BWT-based haplotype formats~\cite{durbin2014pbwt,boucher-uk-biobank}. Our matrices $G$ come from human chromosomes 22, 21, and 20 (from the 1000 Genomes Project~\cite{1000genomes2015}): individuals as rows, SNPs as columns in genomic order, each cell in $\{0,1,2\}$; adjacent columns are in heavy linkage disequilibrium (LD), with repetitiveness that RePair exploits. We use two scales (a $10^5$-variant subset and the full per-chromosome sequence), plus synthetic panels from the coalescent with recombination via \texttt{msprime}~\cite{msprime2016,msprime2022}. The recombination-to-mutation ratio of these synthetic panels tunes LD block structure (reproducible via fixed seeds), with low ratios creating highly compressible, long-LD blocks and high ratios the inverse. We avoid full biobank cohorts (\cref{sec:limitations}); our largest synthetic panel (\texttt{crossover\_synth}, $10{,}000$ individuals, $700{,}000$ variants, $1.00$ G non-zeros) serves as a scale benchmark.

\Cref{tab:geno_through} gives the grammars' structural parameters ($|\mathcal{R}|$, depth $L$, max frontier width $w^{*}$, pass-through inflation $+\text{pt}$). \Cref{tab:geno_time} gives the average time per right product (over 100 vectors) for our GPU streaming engine (\textbf{GPU}); a parallel \textbf{OpenMP} CPU sweep ($20$ threads) and a \textbf{sequential} CPU one; the original \texttt{mm-repair} (\texttt{mmr}) at 1 and 16 threads; and vendor \textbf{cuSPARSE} SpMV, with $\times$GPU the engine's speed-up over cuSPARSE (cuSPARSE time over GPU time). \Cref{tab:geno} adds device footprint and energy versus cuSPARSE.

\begin{table}[t]
\caption{Structural figures, measured on the GB10 node. $|\mathcal{R}|$: RePair non-terminals; $L$: grammar depth; $w^{*}$: maximum live width after completion (the streaming-buffer size); $+\text{pt}$: pass-through nodes added by completion. $\text{K}=10^{3}$, $\text{M}=10^{6}$.}
\label{tab:geno_through}
\centering\footnotesize
\setlength{\tabcolsep}{4pt}
\begin{tabular}{lcccc}
\toprule
matrix & $|\mathcal{R}|$ & $L$ & $w^{*}$ & $+\text{pt}$\\
\midrule
\multicolumn{5}{l}{\textbf{Real Genotypes (1000 Genomes)}} \\
Chr22 & 437.6K & 20 & 146.0K & 668.1K\\
Chr22 full & 4.29M & 24 & 1.41M & 7.02M\\
Chr21 & 404.6K & 24 & 133.1K & 568.6K\\
Chr21 full & 4.30M & 26 & 1.39M & 7.29M\\
Chr20 & 438.3K & 22 & 144.7K & 718.9K\\
Chr20 full & 6.65M & 24 & 2.18M & 11.12M\\
\multicolumn{5}{l}{\textbf{Synthetic Genotypes (Haplotypes)}} \\
\texttt{synth\_small} & 625.3K & 25 & 225.0K & 1.23M\\
\texttt{synth\_large} & 3.94M & 29 & 1.39M & 8.29M\\
\texttt{synth\_ld\_high} & 1.57M & 30 & 498.4K & 3.55M\\
\texttt{synth\_ld\_low} & 3.47M & 21 & 1.54M & 4.85M\\
\texttt{synth\_ind\_large} & 1.40M & 30 & 489.8K & 3.08M\\
\bottomrule
\end{tabular}

\end{table}

\paragraph{Depth and completion.}

The critical depth $L$ remains low across all test configurations: 20 to 30 tiers for the grammars of \cref{tab:geno_through}, and $L=32$ for the larger \texttt{crossover\_synth}. Crucially, the maximum streaming frontier width $w^{*}$ stays well below the base rule count $|\mathcal{R}|$ (about one third on the real chromosomes, up to one half on the synthetics), enabling the dual buffers to fit within device memory. Completion is not free: it adds $1.4$ to $2.3\times|\mathcal{R}|$ pass-through nodes, more than the RePair rules themselves (\cref{tab:geno_through}), i.e., $28$--$53\%$ of the engine's footprint but only $6$--$11\%$ of CSR's. What completion buys is structural: every launch reads only the previous frontier, so the transient state is $2w^{*}$ values per right-hand vector, known before the run, instead of one per rule ($|\mathcal{R}|$) as without completion; this term grows with the batch $B$, the pass-through nodes do not (reducing them is future work). As the genotype alphabet is discrete and small ($\{0,1,2\}$), terminal evaluation is cheap and does not bound performance, leaving traversal and emission atomics as the likely dominant kernel cost.

\paragraph{Time vs.\ cuSPARSE and \texttt{mm-repair}.}

Against the heavily optimized vendor cuSPARSE CSR kernel ($\times$GPU in \cref{tab:geno_time}), our single-vector engine maintains competitive parity on the real chromosomes, from $1.41\times$ faster to $1.20\times$ slower than cuSPARSE. Further, our engine outperforms cuSPARSE on synthetic configurations characterized by high linkage disequilibrium or expanded sample sizes (achieving up to a $4.2\times$ performance improvement on \texttt{synth\_ind\_large}), where the underlying compression factor is at its maximum. Our GPU streaming engine also outperforms the 16-thread CPU \texttt{mm-repair} baseline on every instance, by $1.2\times$ to $11.4\times$.

\begin{table*}[t]
\caption{Genotype average time per right matrix--vector product (ms/vector) on the GB10 node. $\times$GPU: speed-up of the engine over cuSPARSE, i.e., cuSPARSE time over GPU time (above $1$: the engine is faster).}
\label{tab:geno_time}
\centering\footnotesize
\begin{tabular}{lrrrrrrr}
\toprule
matrix & GPU & OpenMP & seq. & mmr (seq) & mmr (16th) & cuSPARSE & $\times$GPU\\
\midrule
\multicolumn{8}{l}{\textbf{Real Genotypes (1000 Genomes)}} \\
Chr22 ($2504\times0.10\text{M}$) & 0.50 & 11.01 & 9.37 & 2.80 & 3.00 & 0.64 & 1.30\\
Chr22 ($2504\times1.06\text{M}$) & 6.69 & 32.03 & 93.22 & 68.60 & 27.80 & 6.21 & 0.93\\
Chr21 ($2504\times0.10\text{M}$) & 0.50 & 14.07 & 10.08 & 2.70 & 0.60 & 0.71 & 1.41\\
Chr21 ($2504\times1.05\text{M}$) & 6.38 & 30.05 & 87.19 & 66.60 & 28.00 & 6.81 & 1.07\\
Chr20 ($2504\times0.10\text{M}$) & 0.48 & 18.48 & 8.92 & 3.00 & 2.70 & 0.68 & 1.40\\
Chr20 ($2504\times1.74\text{M}$) & 11.73 & 40.67 & 139.31 & 139.70 & 42.90 & 9.76 & 0.83\\
\multicolumn{8}{l}{\textbf{Synthetic Genotypes (Haplotypes)}} \\
$\text{synth\_small } (2\text{K}\times50\text{K})$ & 0.51 & 12.22 & 7.95 & 2.80 & 2.50 & 0.82 & 1.59\\
$\text{synth\_large } (5\text{K}\times200\text{K})$ & 2.97 & 31.33 & 50.57 & 46.00 & 26.00 & 7.41 & 2.50\\
$\text{synth\_ld\_high } (5\text{K}\times100\text{K})$ & 0.98 & 23.46 & 18.12 & 7.90 & 11.10 & 3.58 & 3.67\\
$\text{synth\_ld\_low } (5\text{K}\times100\text{K})$ & 3.78 & 23.66 & 66.53 & 47.30 & 21.10 & 3.78 & 1.00\\
$\text{synth\_ind\_large } (10\text{K}\times50\text{K})$ & 0.83 & 22.83 & 16.64 & 9.20 & 6.30 & 3.49 & 4.20\\
\bottomrule
\end{tabular}

\end{table*}

\paragraph{Space and energy vs.\ cuSPARSE.}

\Cref{tab:geno} reports the engine's analytic device footprint, time, and GPU energy versus cuSPARSE CSR SpMV. The engine is consistently smaller and more energy-efficient, especially on highly repetitive or high-LD matrices. We observe a $4.0\times$ to $4.7\times$ footprint reduction across the real chromosomes, which correlates with lower energy usage (up to $35\%$ on the full chromosomes, and roughly half on the $10^5$-variant subsets). This efficiency stems from arithmetic operations running directly on the compressed grammar, without materializing the uncompressed matrix. On synthetics, the LD level sets the scaling: high LD (recombination rate $10^{-9}$) gives a $7.3\times$ smaller space and beats cuSPARSE by about $3.6\times$ in time and $4.6\times$ in energy. Conversely, low LD (recombination rate $10^{-7}$) drops to $3.5\times$ space and runs at time parity. Widening the cohort (\texttt{synth\_ind\_large}, $10{,}000$ individuals) reaches $7.9\times$ smaller space and $4.7\times$ less energy, at the speed-up already reported above. For biobank-scale problems, this memory advantage decides whether a matrix fits in high-speed GPU memory (\cref{fig:space}).

\begin{figure}[t]
\centering
\includegraphics[width=.95\linewidth]{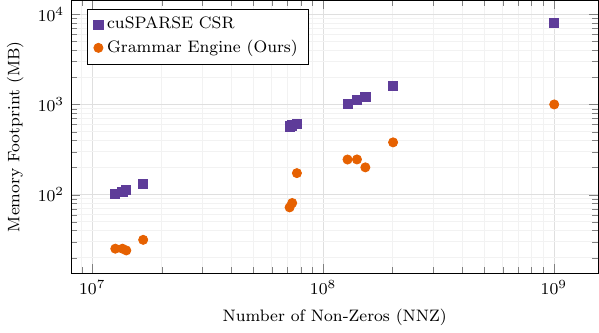}
\caption{Analytic device footprint, grammar engine vs.\ cuSPARSE CSR, on log-log scale; both series are analytic, as in \cref{tab:geno}. Even at the largest scale (\texttt{crossover\_synth}, $1.00$ G non-zeros), cuSPARSE's CSR ($8.03$ GB) still fits, but the Grammar Engine's $1.01$ GB is $8.0\times$ smaller.}
\label{fig:space}
\end{figure}

\paragraph{Pushing to billion-nonzero scale.}

To probe the largest-scale regime, we generated a large synthetic matrix (\texttt{crossover\_synth}, $10{,}000 \times 700{,}000$ genotypes, $1.00$ billion non-zeros). In standard CSR representation, this matrix requires $8.03$ GB of GPU memory; on the GB10's unified $119$ GiB pool, cuSPARSE runs, but at $48.38$ ms per vector and $1976$ mJ/vec: the highest single-vector cost of any genotype matrix we test.

In contrast, our Grammar Engine processes the matrix entirely in its compressed representation, requiring an analytic device footprint of only $1.01$ GB (measured peak: $0.98$ GB), which is $8.0\times$ smaller than CSR. Once constructed, the engine evaluated the right product in $20.09$ ms per vector ($2.41\times$ faster than cuSPARSE and $14.5\times$ faster than the sequential CPU reference), with an energy consumption of $628$ mJ/vec ($3.15\times$ less; see crossover point in \cref{fig:space}). At this scale, the advantage is no longer confined to space but also to time and energy; at a true biobank scale (where the uncompressed matrix exceeds device memory), it is what keeps the matrix resident at all times.

\begin{table*}[t]
\caption{Genotype right product. Single-vector engine vs.\ cuSPARSE CSR SpMV in unified memory on the GB10 node. $\text{nnz}$: non-zeros; $\text{M}=10^{6}$; $\text{K}=10^{3}$. Energy: GPU only (NVML, \cref{sec:setup}).}
\label{tab:geno}
\centering\footnotesize
\begin{tabular}{lr rr rr rr}
\toprule
& & \multicolumn{2}{c}{space (MB)} & \multicolumn{2}{c}{time (ms)} & \multicolumn{2}{c}{energy (mJ)}\\
\cmidrule(lr){3-4}\cmidrule(lr){5-6}\cmidrule(lr){7-8}
matrix ($\text{rows}\times\text{cols}$) & nnz & eng. & cuS. & eng. & cuS. & eng. & cuS.\\
\midrule
\multicolumn{8}{l}{\textbf{Real Genotypes (1000 Genomes)}} \\
Chr22 ($2504\times0.10\text{M}$) & 12.58M & 25.2 & 101.1 & 0.50 & 0.64 & 14 & 26\\
Chr22 ($2504\times1.06\text{M}$) & 127.58M & 246.1 & 1025.4 & 6.70 & 6.21 & 199 & 257\\
Chr21 ($2504\times0.10\text{M}$) & 14.04M & 24.1 & 112.8 & 0.50 & 0.71 & 14 & 31\\
Chr21 ($2504\times1.05\text{M}$) & 140.20M & 246.6 & 1126.4 & 6.40 & 6.81 & 183 & 280\\
Chr20 ($2504\times0.10\text{M}$) & 13.52M & 25.2 & 108.6 & 0.47 & 0.68 & 14 & 28\\
Chr20 ($2504\times1.74\text{M}$) & 201.04M & 381.6 & 1616.1 & 11.72 & 9.76 & 328 & 439\\
\multicolumn{8}{l}{\textbf{Synthetic Genotypes (Haplotypes)}} \\
$\text{synth\_small } (2\text{K}\times50\text{K})$ & 16.62M & 31.6 & 133.2 & 0.49 & 0.82 & 15 & 31\\
$\text{synth\_large } (5\text{K}\times200\text{K})$ & 152.44M & 201.4 & 1221.0 & 2.99 & 7.41 & 96 & 279\\
$\text{synth\_ld\_high } (5\text{K}\times100\text{K})$ & 73.47M & 81.0 & 588.5 & 0.99 & 3.58 & 32 & 146\\
$\text{synth\_ld\_low } (5\text{K}\times100\text{K})$ & 77.01M & 173.8 & 616.8 & 3.74 & 3.78 & 123 & 144\\
$\text{synth\_ind\_large } (10\text{K}\times50\text{K})$ & 71.60M & 72.4 & 573.4 & 0.85 & 3.49 & 29 & 135\\
$\text{crossover\_synth } (10\text{K}\times700\text{K})$ & 1002.97M & 1005.6 & 8030.5 & 20.09 & 48.38 & 628 & 1976\\
\bottomrule
\end{tabular}

\end{table*}

\paragraph{Batched evaluation (SpMM).}

Batching $B$ right-hand vectors ($Y=MX$, the kernel of block power/Krylov iteration~\cite[\S7.3, \S10.1, \S10.3.6]{golub-van-loan-matrix-book}) amortizes grammar traversal. With both sides given their best configuration (the engine's best $B$, and cuSPARSE's best of \texttt{ALG\_DEFAULT}/\texttt{CSR\_ALG2}/\texttt{CSR\_ALG3}), cuSPARSE leads by $1.7\times$ to $10.9\times$. This gap is narrowest on the full real chromosomes ($1.7$ to $2.4\times$) and widest on the two least compressible synthetics, \texttt{synth\_small} and \texttt{synth\_ld\_low} ($7.2\times$ and $10.9\times$; highest $|\mathcal{R}|/\text{nnz}$, \cref{tab:geno_through,tab:geno}); all other instances fall between $3.1\times$ and $4.9\times$. Unlike the single-vector case, the footprint advantage does not persist at large $B$, where the engine's $B$-wide terminal array and streaming buffers overtake cuSPARSE's $B$-invariant CSR (\cref{sec:limitations}). \Cref{tab:geno_spmm} gives the per-dataset times; \texttt{ALG\_DEFAULT} and \texttt{CSR\_ALG2} degrade severely as $B$ grows on these short, ultra-wide shapes, so \texttt{CSR\_ALG3} is our vendor baseline (full sweep in the artifact).

\begin{table*}[t]
\caption{Batched right product (SpMM, $Y=MX$): best per-vector time (ms/vector, averaged over $40$--$50$ iterations) on the GB10 node, with the best batch $B$ of each side (\texttt{a3}: \texttt{CSR\_ALG3}); $\times$cuS: speed-up of cuSPARSE over the engine, i.e., engine time over cuSPARSE time (above $1$: cuSPARSE is faster).}
\label{tab:geno_spmm}
\centering\small
\begin{tabular}{lrrr}
\toprule
matrix & engine ($B$) & cuSPARSE (alg,$B$) & $\times$cuS\\
\midrule
\multicolumn{4}{l}{\textbf{Real Genotypes (1000 Genomes)}} \\
Chr22 ($2504\times0.10\text{M}$) & 0.101 (32) & 0.024 (a3,256) & 4.2\\
Chr22 ($2504\times1.06\text{M}$) & 1.33 (64) & 0.550 (a3,32) & 2.4\\
Chr21 ($2504\times0.10\text{M}$) & 0.096 (16) & 0.020 (a3,256) & 4.9\\
Chr21 ($2504\times1.05\text{M}$) & 1.30 (128) & 0.579 (a3,32) & 2.2\\
Chr20 ($2504\times0.10\text{M}$) & 0.104 (16) & 0.025 (a3,256) & 4.2\\
Chr20 ($2504\times1.74\text{M}$) & 2.13 (256) & 1.25 (a3,32) & 1.7\\
\multicolumn{4}{l}{\textbf{Synthetic Genotypes (Haplotypes)}} \\
$\text{synth\_small } (2\text{K}\times50\text{K})$ & 0.141 (16) & 0.019 (a3,256) & 7.2\\
$\text{synth\_large } (5\text{K}\times200\text{K})$ & 1.04 (32) & 0.242 (a3,64) & 4.3\\
$\text{synth\_ld\_high } (5\text{K}\times100\text{K})$ & 0.380 (32) & 0.121 (a3,128) & 3.1\\
$\text{synth\_ld\_low } (5\text{K}\times100\text{K})$ & 0.945 (256) & 0.087 (a3,128) & 10.9\\
$\text{synth\_ind\_large } (10\text{K}\times50\text{K})$ & 0.334 (16) & 0.079 (a3,256) & 4.2\\
\bottomrule
\end{tabular}

\end{table*}

\subsection{Repetitive graphs: Boolean and Tropical semirings}
\label{sec:graph-exp}

The generalization of \cref{sec:semiring} reinforces our primary thesis: the same memory optimization that justifies grammar compression for genotypes (storing and operating on a highly compressed representation when uncompressed data structures exceed device limits) applies to any highly-repetitive graph topology. To demonstrate this capability, we evaluate the Boolean and Tropical semirings on five adjacency relation matrices built from the Wikidata dataset released with~\cite{time-space-rpq} (\href{https://doi.org/10.5281/zenodo.7254968}{doi:10.5281/zenodo.7254968}), chosen to span scales, edge densities, and topologies: \texttt{wd\_sports\_team} (\texttt{P54}, member of sports team, $332{,}121 \times 29{,}854$), \texttt{wd\_cast\_member} (\texttt{P161}, cast member, $173{,}977 \times 144{,}095$), \texttt{wd\_citizenship} (\texttt{P27}, country of citizenship, $2{,}874{,}250 \times 2{,}556$), \texttt{wd\_occupation} (\texttt{P106}, occupation, $3{,}459{,}933 \times 10{,}610$), and \texttt{wd\_subclass\_of} (\texttt{P279}, subclass of, $1{,}487{,}709 \times 73{,}417$). 

The first three are bipartite, entity-to-entity relations with varying widths and degrees of repetition. The last two add the largest edge set and an ontology-hierarchy (rather than bipartite) topology. Because Wikidata subjects \cite{wikidata-survey,wikidata-cacm,wikidata-making} number in the millions, these matrices are far too large to materialize densely; we thus build them directly from the sparse per-source edge list into the CSRV/RePair grammar without ever forming the uncompressed matrix. Dataset characteristics and grammar structure (averaged over 50 iterations, CPU references verified bit-for-bit) are centralized in \cref{tab:graph_struct}, and runtime/space results in \cref{tab:graph}. We then push the space argument to the graph scale (\cref{tab:graph_scale}) for the two largest Wikidata relations and for a billion-edge Software Heritage graph \cite{swh_dag}.

Compressibility tracks per-source repetition, as with the genotype matrices of \cref{sec:geno-exp}. RePair compresses \texttt{wd\_citizenship}'s $3.06$M edges to just $3{,}295$ rules and the ontology hierarchy \texttt{wd\_subclass\_of}'s $2.02$M edges to $6{,}825$ rules. Both are relations in which each source shares its target structure with many others. \texttt{wd\_sports\_team} and \texttt{wd\_occupation} compress moderately well ($70{,}158$ and $41{,}544$ rules from $1.14$M and $4.60$M edges respectively), while \texttt{wd\_cast\_member} (where each film has a largely distinct cast) compresses least ($51{,}677$ rules from $1.03$M edges, the only relation whose device footprint fails to beat CSR). Grammar depth ranges $L=5$ to $8$ and pass-through completion overhead ranges $11\%$ to $21\%$ across all five relations (\cref{tab:graph_struct}). Operating directly on these compressed DAG topologies, the engine's analytic device footprint undercuts cuSPARSE CSR on four of the five relations ($21\%$ to $33\%$ smaller on \texttt{wd\_sports\_team}, \texttt{wd\_citizenship}, \texttt{wd\_occupation}, and \texttt{wd\_subclass\_of}). Yet, it is \emph{larger} on the poorly repetitive \texttt{wd\_cast\_member} ($11.12$ MB vs.\ $10.24$ MB). This matches the failure mode seen on badly-compressing genotype instances (\cref{sec:geno-exp}): once a relation lacks high repetitiveness, the per-rule bookkeeping overhead of the layered grammar outweighs the modest structural savings. Single-vector times remain within a small factor of cuSPARSE (Boolean $0.09$ to $0.73$ ms vs.\ $0.06$ to $0.50$ ms) while outrunning the sequential CPU reference by $17\times$ to $34\times$. The batched $B{=16}$ sweeps amortize the traversal further (e.g., Boolean \texttt{wd\_sports\_team} drops from $0.092$ to $0.051$ ms/vector; \cref{tab:graph}). These results confirm that the space advantage is real but \emph{conditional on genuine per-source repetitiveness}: a measurable structural fact about each relation.

\begin{table}[t]
\caption{Structural figures for the five Wikidata relations, measured on the GB10 node (columns as in \cref{tab:geno_through}; total: layered rules after pass-through completion).}
\label{tab:graph_struct}
\centering\scriptsize
\setlength{\tabcolsep}{3pt}
\begin{tabular}{lrrrrrr}
\toprule
relation & nnz & $|\mathcal{R}|$ & total & $L$ & $w^{*}$ & $+\text{pt}$\\
\midrule
\texttt{wd\_sports\_team} & 1.14M & 70.2K & 80.9K & 6 & 49.3K & 10.7K\\
\texttt{wd\_cast\_member} & 1.03M & 51.7K & 62.3K & 8 & 42.4K & 10.7K\\
\texttt{wd\_citizenship} & 3.06M & 3.3K & 3.7K & 5 & 2.4K & 432\\
\texttt{wd\_occupation} & 4.60M & 41.5K & 46.7K & 8 & 19.5K & 5.1K\\
\texttt{wd\_subclass\_of} & 2.02M & 6.8K & 7.6K & 5 & 5.8K & 776\\
\bottomrule
\end{tabular}

\end{table}

\paragraph{Semiring-native baseline (GraphBLAS).}

For these semirings, we compare to SuiteSparse:GraphBLAS~\cite{graphblas} running the same operation (\texttt{lor\_land}/\texttt{min\_plus} \texttt{mxv}) on the \emph{uncompressed} matrix (20 threads, \cref{tab:graph}). Working on the \emph{compressed} grammar, our GPU engine outperforms GraphBLAS on \emph{all five} Boolean sweeps, though narrowly ($1.1\times$--$1.2\times$) on \texttt{wd\_cast\_member} and \texttt{wd\_occupation}. For the Tropical semiring, it is $3\times$ to $7\times$ faster across all five relations (e.g., \texttt{wd\_occupation}: $0.73$ vs.\ $3.57$ ms). A GPU-native cuGraph BFS/SSSP reference is discussed in \cref{sec:limitations}; it performs a \emph{full} traversal rather than a single mat-vec, serving as an end-to-end comparison rather than a per-mat-vec baseline.

\paragraph{Scaling to 10M--166M edges.}

To probe the space argument at scale, we build the two largest single Wikidata relations, bracketing the compressibility spectrum: the narrow categorical \texttt{wd\_country} ($10.1$M edges) and the citation network \texttt{wd\_cites\_work} ($166.7$M edges); see \cref{tab:graph_scale}. Their dense forms ($2.2$ GB, $10.8$ TB) are unusable, making the sparse-to-grammar construction (\cref{sec:graph-exp}) essential to make them addressable. 

Two effects stand out. Compressibility tracks per-source repetition: \texttt{wd\_country} collapses to $2{,}002$ rules (serializing $\mathbf{16.8\times}$ below CSR's footprint), whereas \texttt{wd\_cites\_work} barely compresses ($3.1\times$) due to its largely unique citation lists. More subtly, the \emph{device-resident} advantage is far smaller than the serialized one. The sweep still materializes one root per row plus the persistent terminal array, a $\Theta(\text{rows})$ cost. Consequently, even the $16.8\times$-compressible \texttt{wd\_country} is only $\sim\!25\%$ smaller than CSR in device bytes, and \texttt{wd\_cites\_work} is essentially tied. At graph scale, the grammar's decisive advantage lies in its \emph{serialized/transfer} size and in avoiding the dense matrix construction, rather than device-resident bytes, which are lower-bounded by the per-row root/terminal structure. This marks a boundary complementary to the genotype scale regime (\cref{sec:geno-exp}), where the engine's device footprint stays several-fold below CSR even as both fit.

\paragraph{A billion-edge software graph (Software Heritage).}

Our largest instance, and the only one from the \emph{software} domain, is drawn from the Software Heritage (SWH) archive~\cite{swh_dag}, the global Merkle-DAG of publicly archived source code \cite[\S2.1]{swh-ecosystems-book}\cite{swh_cacm,di-cosmo-archiving}. It is repetitive by construction. Source code evolves by small edits over many revisions and forks that reuse the same files and directory subtrees across the archive: exactly the redundancy RePair exploits. From the \texttt{2021-03-23-popular-3k-python} export, we build a single boolean adjacency of $45.7$M nodes and $1.22$ billion edges (\cref{tab:graph_scale}). Its dense form ($\approx\!261$ TB) cannot be materialized, and even cuSPARSE's CSR needs $\approx\!10$ GB. RePair compresses it to a $30.3$M-rule grammar serialized to $490$ MB, representing a $21\times$ reduction under that CSR. 

Crucially, and unlike \texttt{wd\_cites\_work}, the advantage here \emph{also} survives device-resident ($2.79$ GB, $3.7\times$ under CSR), because SWH is simultaneously highly repetitive \emph{and} dense enough ($\approx\!27$ edges/row) to amortize the $\Theta(\text{rows})$ root/terminal cost; the engine also beats cuSPARSE in time there ($52.8$ vs.\ $65.4$ ms). Structurally, its grammar is deep and wide ($L=70$, $w^{*}=8.98$M), presenting a distinct stress test compared to the shallow bipartite relations ($L=5$ to $8$). Our grammar serializes to $\approx\!3.2$ bits/edge; WebGraph~\cite{webgraph,webgraph_rust} stores it slightly smaller ($2.485$ bits/edge) but does not evaluate a semiring mat-vec on the compressed form on a GPU.

\begin{table*}[t]
\caption{Largest graph relations that bracket the compressibility spectrum: two Wikidata relations and a billion-edge Software Heritage software graph. REANS: serialized grammar (the engine's on-disk/transfer size); CSR / eng.\ dev.: analytic device footprint of cuSPARSE CSR and the engine (for these Boolean relations CSR's serialized size and device footprint coincide up to the negligible $x/y$ vectors, so the CSR column doubles as both baselines); Bool: single-vector Boolean time (engine / cuSPARSE), GB10 node. Dense forms ($2.2$ GB / $10.8$ TB / $261$ TB) are unusable.}
\label{tab:graph_scale}
\centering\footnotesize
\begin{tabular}{lrrrrrrr}
\toprule
relation & nnz & $|\mathcal{R}|$ & $L$ & REANS & CSR & eng.\ dev. & Bool eng/cuS\\
 & & & & (MB) & (MB) & (MB) & (ms)\\
\midrule
\texttt{wd\_country} ($10.1$M) & 10,089,284 & 2,002 & 5 & 9.6 & 161.2 & 120.8 & 1.63 / 1.41\\
\texttt{wd\_cites\_work} ($166.7$M) & 166,682,725 & 9,919,927 & 15 & 459.0 & 1439.7 & 1478.4 & 42.71 / 34.56\\
\texttt{swh} ($1.22$G) & 1,218,488,928 & 30,292,590 & 70 & 490.1 & 10301.0 & 2792.9 & 52.79 / 65.42\\
\bottomrule
\end{tabular}

\end{table*}

\begin{table*}[t]
\caption{Graph right product under Boolean and Tropical semirings across five Wikidata relations on the GB10 node (single-vector time, slash-separated per method: Boolean CSR/eng/GB, Tropical eng/GB; batched engine-only at $B{=}16$; footprint columns analytic CSR/eng). CSR: uncompressed cuSPARSE (Tropical: no vendor kernel, N/A). GB: SuiteSparse:GraphBLAS, 20 threads (\texttt{lor\_land}/\texttt{min\_plus} \texttt{mxv}). Structure in \cref{tab:graph_struct}.}
\label{tab:graph}
\centering\footnotesize
\begin{tabular}{lrrr rr rr}
\toprule
& & & analytic MB & \multicolumn{2}{c}{Boolean} & \multicolumn{2}{c}{Tropical}\\
\cmidrule(lr){5-6}\cmidrule(lr){7-8}
relation & nnz & $|R|$ & (CSR/eng) & single (CSR/eng/GB) & $B{=}16$ (eng) & single (eng/GB) & $B{=}16$ (eng) \\
\midrule
\texttt{wd\_sports\_team} & 1,136,249 & 70,158 & 11.87/9.33 & 0.064/0.092/0.267 & 0.051 & 0.092/0.507 & 0.051 \\
\texttt{wd\_cast\_member} & 1,033,124 & 51,677 & 10.24/11.12 & 0.056/0.116/0.144 & 0.046 & 0.125/0.383 & 0.059 \\
\texttt{wd\_citizenship} & 3,063,058 & 3,295 & 47.52/34.63 & 0.374/0.501/0.684 & 0.310 & 0.438/3.021 & 0.280 \\
\texttt{wd\_occupation} & 4,596,658 & 41,544 & 64.51/43.35 & 0.501/0.729/0.818 & 0.390 & 0.735/3.570 & 0.395 \\
\texttt{wd\_subclass\_of} & 2,024,347 & 6,825 & 28.40/19.65 & 0.173/0.222/0.550 & 0.156 & 0.247/1.778 & 0.179 \\
\bottomrule
\end{tabular}

\end{table*}

\subsection{Correctness}
\label{sec:correctness}

Every algorithm and configuration we report is checked on a shared input vector against independent implementations, which must agree: the host CPU reference sweep (sequential and 20-thread OpenMP), the \texttt{mm-repair} CPU grammar mat-vec, the materialized CSR format (cuSPARSE for $(+,\times)$ and Boolean, plus a CPU SpMV / min-plus over the \emph{same} grammar-reconstructed matrix), the batched (SpMM) engine and cuSPARSE kernels, and, for the graph semirings, SuiteSparse:GraphBLAS ($\lor\!\land$ and $\min,+$). Across all genotype and Wikidata relations, these agree within float precision for $(+,\times)$ and bit-for-bit for Boolean and Tropical; the tested Software Heritage graph ($45.7$M nodes, $1.22$ billion edges) is verified bit-for-bit against both the CPU reference sweep and the materialized CSR.

\section{Limitations}
\label{sec:limitations}

We delineate the main boundaries of the current evaluation.

\begin{enumerate}
    \item \emph{Evaluation scale.} Full biobank-scale matrices (hundreds of thousands of individuals) are our target but are modeled here by synthetics up to $10\text{K}$ samples. Whether pass-through inflation and buffer sizing remain stable at full scale without hurting occupancy remains to be verified. Mitigating this, our synthetic panels are not ad-hoc: they are simulated with \texttt{msprime}~\cite{msprime2016,msprime2022}, the community-standard coalescent-with-recombination simulator, whose panels reproduce the haplotype-block and LD structure of real cohorts. What our synthetics leave untested is thus the scale of that structure, not its realism.
    \item \emph{Host-side construction cost.} The compressed format carries a one-time offline cost (the RePair build, plus completion and level-bucketing, at most $12.4$ s on SWH), amortized after a few hundred products (per-matrix figures in the artifact) and thus negligible in the many-vector regimes we target.
    \item \emph{Batched throughput.} cuSPARSE's tuned SpMM outperforms our engine across all configurations ($1.7\times$ to $10.9\times$). The single-vector memory advantage also erodes at large batch: the engine holds its terminal array (one slot per distinct leaf $\langle\ell,j\rangle$, $|T|$ in all) and both streaming buffers $B$-wide, a cost proportional to $(|T|+2w^{*})B$ absent from cuSPARSE, whose CSR is $B$-invariant and grows only with the dense $X,Y$. On \texttt{Chr22 full} the footprints cross near $B\approx50$, below the engine's throughput-optimal $B{=}64$, so there cuSPARSE is smaller; evaluating the leaves in the sweep rather than materializing them $B$-wide would drop the $\Theta(|T|B)$ term, which we leave to future work. What survives is structural: we never materialize the uncompressed matrix.
    \item \emph{Graph analysis.} The space advantage over CSR is conditional on the relation RePair-compressing (\cref{tab:graph_struct,tab:graph}) and reverses on poorly repetitive ones (\texttt{wd\_cast\_member}); our GraphBLAS comparison is a single \texttt{mxv}; a fixpoint-iterated, end-to-end evaluation against cuGraph BFS/SSSP is left to future work.
    \item \emph{Single unified-memory platform.} All measurements come from one GB10 node, whose CPU and GPU share a single coherent physical memory (\cref{sec:setup}). We do not isolate the on-demand paging cost of managed allocations (page faults and mappings): the engine's per-vector averages include it, amortized over the timed loop, whereas cuSPARSE's exclude it after a warm-up run, a bias against the engine. Independent artifact-evaluation runs reproduced our timings within a few percent on another GB10, and our structural and memory results exactly on discrete GPUs, where however the engine was slower than cuSPARSE, increasingly with memory bandwidth (on average $1.45\times$ on a V100, $2.5\times$ on an H200), and much slower with managed memory on an RTX~4090 under WSL2, where the cause was not isolated; our time and energy claims are thus specific to the GB10. The analytic footprints and structural metrics ($L$, $w^{*}$, pass-through counts) are architecture-independent; timings and energy also vary with OS scheduling.
\end{enumerate}

\section{Conclusion and future work}
\label{sec:conclusion}

We showed that a right SpMV product over a grammar-compressed matrix is a conflict-free, level-synchronous DAG gather, and that the efficient parallel model is \emph{properly layered}: the sweep streams through two alternating buffers, freeing frontiers immediately (\cref{lem:liveness}). We maintain RePair's compression and ensure proper layering via post-processing pass-through completion. On large genotype matrices against cuSPARSE CSR, our engine’s primary advantage is \emph{space}. On real chromosomes, it occupies $4.0\times$ to $4.7\times$ less device memory at single-vector parity ($1.41\times$ faster to $1.20\times$ slower) and lower energy, scaling to $7.9\times$ smaller and $4.2\times$ faster on highly compressible synthetics, while matching large-scale polygenic scoring ($y=G\beta$). Requiring only an associative and commutative combine, our engine extends this space advantage to repetitive graphs with up to a billion edges. Extensions include: scaling to full biobank cohorts where CSR exceeds GPU memory; profiling on discrete-memory GPUs, including managed-memory page migration; extending the semiring framework (\cref{sec:semiring}) to fixpoint BFS/SSSP against GraphBLAS/cuGraph; left multiplication ($x^{\mathsf{T}}=y^{\mathsf{T}}M$) via pull-based reduction; integrating height-balancing~\cite{ganardi,balancing-urbina} to shorten $L$; reducing completion inflation via a compact pass-through encoding or minimum-dummy layerings~\cite{gansner_drawing_directed_graphs,healy_layer_dag}, adapted to shared pass-throughs, trading depth for fewer pass-throughs (bounding width and height is NP-hard on general DAGs~\cite{healy_layer_dag}); exploring recompression-based constructions~\cite{recompression}; and targeting FPGA pipelines.

\section*{Acknowledgments}

Francesco Tosoni was supported by the Alfred P.~Sloan Foundation with the grant \#\href{https://sloan.org/grant-detail/g-2025-25193}{G-2025-25193} (\url{sloan.org}). Francesco Tosoni thanks Travis Gagie for discussions on genotype datasets and semirings. The authors thank the anonymous ALENEX~2027 reviewers for their insightful comments, and the Artifact Evaluation Committee for their thorough evaluation and valuable discussions.

\bibliographystyle{siamplain}
\bibliography{references}

\end{document}